\documentclass[runningheads]{llncs}
\usepackage[T1]{fontenc}
\usepackage{amsmath,amssymb,mathtools}
\usepackage{booktabs}
\usepackage{xcolor}
\usepackage{float}
\usepackage{microtype}
\usepackage[hidelinks]{hyperref}

\newcommand{\R}{\mathbb{R}}
\newcommand{\Pcal}{\mathcal{P}}

\newcommand{\e}{\mathbf{e}}
\newcommand{\dd}{\,\mathrm{d}}
\newcommand{\norm}[1]{\lVert #1\rVert}
\newcommand{\N}{\mathcal{N}}

\newcommand{\KL}{\operatorname{KL}}

\title{Compositional Boundaries for Density Fusion\thanks{To appear in the Proceedings of the 17th International Conference on Scalable Uncertainty Management (SUM 2026), LNAI, Athens, Greece}}
\titlerunning{Compositional Boundaries for Density Fusion}

\author{Ratan Bahadur Thapa\inst{1}\thanks{Corresponding author: ratan.thapa@ki.uni-stuttgart.de} \and Ali Darijani\inst{3,4} \and Jürgen Beyerer\inst{3,4} \and \\ Steffen Staab\inst{1,2}}
\authorrunning{Thapa et al.}
\institute{Institute for Artificial Intelligence, University of Stuttgart, Stuttgart, Germany\\
\and
 University of Southampton, United Kingdom\\
 \and
Institute for Anthropomatics and Robotics, Karlsruhe Institute of Technology, \\
 \and
Fraunhofer Institute of Optronics, System Technologies and Image Exploitation, Karlsruhe, Germany
}

\begin{document}
\maketitle
\begin{abstract}
Distributed uncertainty-management systems often combine local probabilistic
models along aggregation trees chosen by communication, privacy, or scheduling
constraints. The final density should depend on the weighted sources, not on the
particular order in which intermediate nodes combine them. We study this
requirement as an algebraic compositionality problem for binary fusion of
weighted probability densities. The central question is when a local fusion rule
can be executed hierarchically while remaining order-invariant. 
We establish a compositional boundary for local segment-valued fusion rules.
Within the class of continuous binary rules with additive output weights and
weight-only coefficients, order-invariant hierarchical execution characterizes
normalized weighted linear pooling; norm-induced segment balancing realizes the
corresponding coefficient. Smooth endpoint-to-candidate
$f$-divergence balancing has a different local geometry: its quadratic expansion
induces square-root effective weights, showing why pairwise solvability alone is
insufficient for schedule-independent fusion. We show that this obstruction is local to endpoint-to-candidate binary balancing,
whereas global divergence barycenters retain additive-weight local limits. 
Finally, Gaussian mixtures show how the same issue
appears in finite model classes: exact fusion is compositional, whereas stepwise
compression is compositional only under a congruence condition on unnormalized
component measures. These results distinguish exact schedule-independent fusion
from global aggregation objectives and local approximation heuristics.
\end{abstract}

\section{Introduction}\label{sec:introduction}
Scalable uncertainty management often requires probabilistic predictions to be
combined across distributed sources. Edge devices estimate local densities,
hospitals train models under privacy constraints, and federated systems exchange
summaries rather than raw data. This makes aggregation more than probability
averaging: it requires fusion of probabilistic evidence with source-dependent
weights. We study weighted probability densities $(w,p)$, where $p$ is a
density and $w>0$ represents sample size, reliability, confidence, or priority.

The central semantic requirement is that the final density depend only on the
weighted sources, not on the aggregation tree. If the same weighted sources are combined along a left-deep tree, a balanced tree, or a dynamically scheduled aggregation tree, the final density should not change merely because the communication pattern changed. For states $x_1,x_2,x_3$, the computations $(x_1\oplus x_2)\oplus x_3$ and $x_1\oplus(x_2\oplus x_3)$ should have the same meaning; with permuted arrivals, $x_1\oplus x_2$ and $x_2\oplus x_1$ should agree as well. Thus, the binary fusion operation should form a commutative monoid on weighted density states. This requirement is stronger than pairwise solvability and stronger than closure of a model class under convex mixtures \cite{grabisch2009aggregation}.

A natural local protocol chooses a dissimilarity, solves a balanced binary fusion problem at each internal node, and repeats this step through the tree. Pairwise solvability is not enough: Kullback--Leibler (KL) segment balancing~\cite{kullback1951information} already gives different Bernoulli densities for the two parenthesizations of three equal-weight sources. The classical schedule-independent candidate is weighted linear pooling ~\cite{stone1961opinion,genest1986combining,jacobs1995methods},
\begin{equation}
  p_{\mathrm{pool}}=\frac{\sum_{i=1}^m w_i p_i}{\sum_{i=1}^m w_i},
  \label{eq:pooling-intro}
\end{equation}
which is standard in opinion aggregation and expert probability assessment. We do not study weighted linear pooling as a one-shot aggregate. Instead, we ask
when a binary fusion rule can be executed at arbitrary internal nodes of an
aggregation tree without making the final result schedule-dependent. Our results
therefore concern algebraic consistency of distributed execution, not statistical
optimality for a particular data-generating model.

Our main contribution is a compositional boundary result for local segment-based
fusion protocols. Norm-induced segment balancing realizes weighted linear
pooling. Conversely, any continuous segment-valued binary rule that propagates
additive weights, chooses its mixing coefficient only from the two input weights,
and is commutative and associative over all convex density classes with three
affinely independent densities must use the normalized weight coefficient.
Therefore, within this operational class, order-invariant hierarchical execution
forces weighted linear pooling. Associative alternatives
exist only after changing the additive quantity carried by the state, for example
by propagating transformed weights.

The main boundary concerns divergences. Under the local bounded-ratio assumptions used in Section~\ref{sec:divergence}, endpoint-to-candidate balancing for every twice continuously differentiable $f$-divergence with positive curvature at equality is locally quadratic along mixture segments, so its coefficient converges to $\sqrt{w_b}/(\sqrt{w_a}+\sqrt{w_b})$, not to $w_b/(w_a+w_b)$. We establish first-order tree-independence under square-root weight propagation and contrast it with global $f$-barycenter objectives, whose local quadratic limiting target is additive-weighted. Thus, the obstruction does not rule out global divergence barycenters, logarithmic pooling, Wasserstein barycenters, or endpoint-dependent optimization rules~\cite{kullback1951information,ali1966general,csiszar1967information}.

Gaussian mixtures give a finite model-class instance of the same issue. Exact weighted fusion stays inside the class of all finite Gaussian mixtures by concatenating components, while fixed-budget subclasses are generally not closed. Stepwise compression is order-invariant exactly when the idempotent compression map is compatible with unnormalized mixture addition, equivalently when it induces a monoid congruence.

The remainder of this paper is organized as follows: In Sections~\ref{sec:balance}--\ref{sec:gmm}, we develop these claims from
local balance to algebraic order-invariance, divergence barriers, and compressed
Gaussian-mixture fusion. In Section~\ref{sec:related}, we compare the results
with opinion pooling, evidence fusion~\cite{shafer1976mathematical},
barycenters~\cite{agueh2011barycenters}, and mixture
reduction~\cite{dempster1977maximum}.

\section{Segment Balance}\label{sec:balance}
Let $(\Omega,\mathcal F,\mu)$ be a measure space. A probability density
function, or PDF, is a measurable function $p:\Omega\to\R_{\geq0}$ with
$\int_\Omega p\dd\mu=1$, and PDFs are identified up to almost-everywhere
equality. A PDF class $\Pcal$ is a nonempty convex set of PDFs in an ambient
real vector space of measurable functions modulo almost-everywhere equality.
For $p,q\in\Pcal$ and $t\in[0,1]$, we write
\begin{equation}
  [p,q]_t=(1-t)p+tq .
  \label{eq:segment}
\end{equation}
A weighted PDF state is a pair $(w,p)$ with $w\in\R_{>0}$ and
$p\in\Pcal$. We also use a formal identity element $\e$ and set
$\widehat{\Pcal}=(\R_{>0}\times\Pcal)\cup\{\e\}$; the element $\e$ is not a
zero-weight density.

When a norm-induced distance is used, we assume that all differences
$p-q$, for $p,q\in\Pcal$, belong to a normed ambient vector space and have
finite norm. A norm-induced distance on $\Pcal$ is a dissimilarity
$\Delta(p,q)=\norm{p-q}$. A divergence on $\Pcal$ is a map
$D:\Pcal\times\Pcal\to\R_{\geq0}\cup\{\infty\}$ such that
\begin{equation}
 D(p,q)=0 \quad\text{if and only if}\quad p=q \text{ almost everywhere}.
 \label{eq:div-separation}
 \end{equation}
We do not require divergences to be symmetric or to
satisfy the triangle inequality.

\begin{definition} \label{def:segment-balance}
Let $\Pcal$ be a PDF class, let
$\Delta:\Pcal\times\Pcal\to\R_{\geq0}\cup\{\infty\}$, and let
$(w_a,p_a),(w_b,p_b)\in\R_{>0}\times\Pcal$. A parameter
$t^*\in[0,1]$ with finite endpoint dissimilarities is a segment-balancing point
for these two weighted states with respect to $\Delta$ if
\begin{equation}
  w_a\Delta(p_a,[p_a,p_b]_{t^*})
  =
  w_b\Delta(p_b,[p_a,p_b]_{t^*}).
  \label{eq:balance}
\end{equation}
The associated fused PDF is $[p_a,p_b]_{t^*}$.
\end{definition}

The definition restricts the binary output to the mixture segment and balances
the weighted dissimilarities to the two endpoints. For norm-induced distances,
this local condition already determines the mixing coefficient.

\begin{theorem}
\label{thm:norm-characterization}
Let $\Delta(p,q)=\norm{p-q}$ be norm-induced. If $p_a\neq p_b$, then the unique segment-balancing point is
\begin{equation}
  t^*=\frac{w_b}{w_a+w_b}
  \quad\text{and}\quad
  [p_a,p_b]_{t^*}
  =
  \frac{w_a}{w_a+w_b}p_a+
  \frac{w_b}{w_a+w_b}p_b.
  \label{eq:norm-solution}
\end{equation}
\end{theorem}

\begin{proof}
Let $p_t=[p_a,p_b]_t$. Then, $p_t-p_a=t(p_b-p_a)$ and $p_t-p_b=(1-t)(p_a-p_b)$, so absolute homogeneity gives $\norm{p_t-p_a}=t\norm{p_b-p_a}$ and $\norm{p_t-p_b}=(1-t)\norm{p_b-p_a}$. If $p_a\neq p_b$, cancellation of the positive common factor in the balance equation yields $w_at=w_b(1-t)$, thus~\eqref{eq:norm-solution}; uniqueness follows because the equation is linear in $t$. If $p_a=p_b$, all segment points are the same density.
\end{proof}

When $p_a=p_b$, every parameter produces $p_a$, so the fused density remains unique. The result depends on using a norm rather than its square. If $\Delta(p,q)=\norm{p-q}^2$, the same calculation gives $w_at^2=w_b(1-t)^2$ and hence $t^*=\sqrt{w_b}/(\sqrt{w_a}+\sqrt{w_b})$. 
The same square-root coefficient reappears for smooth divergences in
Section~\ref{sec:divergence}.

The norm-induced rule has the required compositional structure. On $\widehat{\Pcal}$, let $\oplus$ be the operation with $\e\oplus x=x\oplus\e=x$ and, for positive states,
\begin{equation}
  (w_a,p_a)\oplus(w_b,p_b) =
  \left(w_a+w_b,
  \frac{w_ap_a+w_bp_b}{w_a+w_b}
  \right).
  \label{eq:oplus}
\end{equation}

\begin{proposition}\label{prop:monoid}
Let $\eta$ be the map given by
\[\eta(\e)=(0,0)\quad\text{and}\quad \eta(w,p)=(w,wp)\]
for positive states $(w,p)$. Then, $\eta$ represents $\oplus$ by addition:
\begin{equation}
  \eta((w_a,p_a)\oplus(w_b,p_b))
  =
  \eta(w_a,p_a)+\eta(w_b,p_b).
  \label{eq:additive-representation}
\end{equation}
Consequently, $(\widehat{\Pcal},\oplus,\e)$ is a commutative monoid, and 
every nonempty finite family $(w_i,p_i)_{i=1}^m$ of positive weighted states
evaluates, independently of order and parenthesization, to
\begin{equation}
  \left(
  \sum_{i=1}^m w_i,
  \frac{\sum_{i=1}^m w_i p_i}{\sum_{i=1}^m w_i}
  \right).
  \label{eq:finite-fusion}
\end{equation}
\end{proposition}

\begin{proof}
For positive states, applying $\eta$ to~\eqref{eq:oplus} gives $(w_a+w_b,w_ap_a+w_bp_b)=\eta(w_a,p_a)+\eta(w_b,p_b)$. The identity laws hold by definition, closure follows from convexity of $\Pcal$, and commutativity and associativity follow from addition in the represented coordinates. Iteration gives~\eqref{eq:finite-fusion}.
\end{proof}

Proposition~\ref{prop:monoid} implies that each intermediate state need only carry the accumulated weight and weighted density, independently of the aggregation tree.

\section{Algebraic Characterization}\label{sec:algebra}

We prove a converse for the operational class introduced above. For binary
segment-valued rules with additive output weights and weight-only coefficients,
commutative and associative execution determines the mixing coefficient uniquely.

Consider a binary operation $\odot$ on weighted PDF states of the form
\begin{equation}
  (w_a,p_a)\odot(w_b,p_b)
  =
  \left(w_a+w_b,
  (1-\alpha(w_a,w_b))p_a+\alpha(w_a,w_b)p_b
  \right),
  \label{eq:general-rule}
\end{equation}
where $\alpha:\R_{>0}^2\to(0,1)$ is continuous and depends only on the two weights. The assumption that coefficients depend only on weights excludes endpoint-adaptive barycenters and other global optimization rules; those belong to a different design space. Affine independence means that no one of the three densities lies in the affine span of the other two; it is used below only to identify barycentric coefficients.

We use the following standard continuous cocycle representation for the additive semigroup~\cite{aczel2006lectures}.
\begin{lemma}
\label{lem:cocycle-representation}
Let $m:\R_{>0}^2\to(0,1)$ be continuous and satisfy
\begin{equation*}
  m(a,b)m(a+b,c)=m(a,b+c)
\end{equation*}
for all $a,b,c>0$. Then, there is a positive continuous function $F:\R_{>0}\to\R_{>0}$, unique up to multiplication by a positive constant, such that
\begin{equation*}
  m(a,b)=\frac{F(a)}{F(a+b)}.
\end{equation*}
\end{lemma}
\begin{proof}
Let $h(a,b)=-\log m(a,b)$. Then $h$ is continuous and satisfies
$h(a,b)+h(a+b,c)=h(a,b+c)$. For fixed $a,c>0$, this identity implies
$h(a,b)=h(a,b+c)-h(a+b,c)$, and hence $h(a,b)\to0$ as $b\downarrow0$.
Choose $r>0$. Set $H(r)=0$, set $H(x)=h(r,x-r)$ for $x>r$, and set
$H(x)=-h(x,r-x)$ for $0<x<r$. The preceding boundary limit makes $H$
continuous at $r$. The cocycle identity then gives, in the three cases
$a>r$, $a+b<r$, and $a<r<a+b$, the representation
$h(a,b)=H(a+b)-H(a)$; the boundary cases follow by continuity. With
$F=\exp H$ we obtain $m(a,b)=F(a)/F(a+b)$. If $\widetilde F$ is another
positive continuous representative, then
$F(a)/\widetilde F(a)=F(a+b)/\widetilde F(a+b)$ for all $a,b>0$, so this
ratio is constant on $\R_{>0}$ and $F$ is unique up to a positive factor.
\end{proof}

\begin{theorem}
\label{thm:converse}
If $\odot$ is commutative and associative for all choices of PDFs in every convex PDF class containing at least three affinely independent densities, then
\begin{equation}
  \alpha(w_a,w_b)=\frac{w_b}{w_a+w_b}
  \label{eq:forced-alpha}
\end{equation}
for all $w_a,w_b>0$.
\end{theorem}

\begin{proof}
Let $m(a,b)=1-\alpha(a,b)$. Expanding
$((a,p_a)\odot(b,p_b))\odot(c,p_c)$ gives the coefficient of $p_a$ as
$m(a,b)m(a+b,c)$, whereas expanding
$(a,p_a)\odot((b,p_b)\odot(c,p_c))$ gives the coefficient of $p_a$ as
$m(a,b+c)$. Since the three densities are affinely independent, equality of
the two parenthesizations implies
$m(a,b)m(a+b,c)=m(a,b+c)$. Lemma~\ref{lem:cocycle-representation} yields
$m(a,b)=F(a)/F(a+b)$ for a positive continuous $F$. Commutativity gives
$m(a,b)+m(b,a)=1$, hence $F(a+b)=F(a)+F(b)$; continuity gives
$F(w)=cw$ on $\R_{>0}$. Therefore $m(a,b)=a/(a+b)$ and
$\alpha(a,b)=b/(a+b)$.
\end{proof}

The quantification over convex PDF classes with three affinely independent
densities is used only to identify barycentric coefficients in the associativity
argument. Theorem~\ref{thm:converse} is not a uniqueness theorem for all
probabilistic aggregation rules. It applies only to continuous segment-valued
binary rules with additive output weights and weight-only coefficients.
Non-segment-valued rules, logarithmic pooling, endpoint-dependent barycenters,
and global optimization principles fall outside these assumptions.
The theorem identifies the point at which hierarchical execution itself forces the normalized linear-pooling coefficient.

The additive-weight assumption is sharp: supplied weights may first be
transformed into effective weights, which are then propagated additively.

\begin{proposition}\label{prop:effective-weight}
Let $g:\R_{>0}\to\R_{>0}$ be a continuous, strictly increasing bijection. On
$\widehat{\Pcal}$, let $\oplus_g$ be the operation with
$\e\oplus_g x=x\oplus_g\e=x$ for all $x\in\widehat{\Pcal}$ and, for positive states,
\begin{equation}
  (a,p)\oplus_g(b,q) =
  \left(
  g^{-1}(g(a)+g(b)),
  \frac{g(a)p+g(b)q}{g(a)+g(b)}
  \right).
  \label{eq:effective-weight-rule}
\end{equation}
Then, $(\widehat{\Pcal},\oplus_g,\e)$ is a commutative monoid, and every
finite family $(w_i,p_i)_{i=1}^m$ of positive weighted states evaluates to
\begin{equation}
  \left(
  g^{-1}\!\left(\sum_{i=1}^m g(w_i)\right),
  \frac{\sum_{i=1}^m g(w_i)p_i}{\sum_{i=1}^m g(w_i)}
  \right).
  \label{eq:effective-weight-finite}
\end{equation}
\end{proposition}

\begin{proof}
The map $\eta_g(w,p)=(g(w),g(w)p)$, with $\eta_g(\e)=(0,0)$, represents $\oplus_g$ by addition. Since $g$ is a bijection, a represented pair $(s,r)$ with $s>0$ recovers the state as $(g^{-1}(s),r/s)$. Thus, commutativity, associativity, identity, and the finite formula~\eqref{eq:effective-weight-finite} follow from addition.
\end{proof}

With $g(w)=w$, Proposition~\ref{prop:effective-weight} recovers ordinary weighted linear pooling. With $g(w)=\sqrt w$, the effective contribution of a source is proportional to the square root of its supplied weight, and the accumulated output weight is $(\sum_i\sqrt{w_i})^2$. In Proposition~\ref{prop:effective-weight}, the coefficient is weight-only, but it is normalized in the transformed weights $g(a),g(b)$, not in the additively propagated supplied masses $a,b$. Thus, different associative segment rules are possible only after changing which numerical quantity the distributed computation treats as additive.

When weights represent sample masses, weighted linear pooling has the usual source-mixture interpretation: if $\Pr(S=i)=w_i/\sum_jw_j$ and $X\mid S=i$ has density $p_i$, then the marginal density of $X$ is $\sum_iw_ip_i/\sum_iw_i$. Reliability or priority weights require domain-specific calibration, but the algebraic condition for order-invariant execution is unchanged.

\section{Divergence Barriers}\label{sec:divergence}

We next show why smooth information divergences change the compositional behavior. Endpoint-to-candidate segment balancing compares $D(p_a\Vert p_t)$ with $D(p_b\Vert p_t)$ for a segment point $p_t=[p_a,p_b]_t$. Such equations can be solvable pairwise, but they do not generally provide an associative binary operation under additive source weights.

Let $D$ be a divergence satisfying~\eqref{eq:div-separation}. For two weighted states $(w_a,p_a)$ and $(w_b,p_b)$, let
\begin{equation}
  \Phi(t)=w_aD(p_a,[p_a,p_b]_t)-w_bD(p_b,[p_a,p_b]_t).
  \label{eq:phi}
\end{equation}
If $p_a\neq p_b$ and, for $p\in\{p_a,p_b\}$, each scalar map $t\mapsto D(p,[p_a,p_b]_t)$ is finite and continuous on $[0,1]$, then $\Phi(0)<0$ and $\Phi(1)>0$. The intermediate value theorem gives a root in $(0,1)$. This observation proves pairwise existence only; it does not give uniqueness, a closed form, or associativity.

We now prove the obstruction for $f$-divergences. Let $f:(0,\infty)\to\R$ be convex with $f(1)=0$. The associated $f$-divergence is
\begin{equation}
  D_f(p\Vert q)=\int_\Omega q\,f\!\left(\frac{p}{q}\right)\dd\mu,
  \label{eq:f-divergence}
\end{equation}
whenever the integral is defined, with the usual extended-value conventions at points where $q=0$. When we use $D_f$ as a divergence in the sense of~\eqref{eq:div-separation}, we restrict to choices of $f$ and PDF classes for which $D_f(p\Vert q)=0$ implies $p=q$ almost everywhere. The local results below require only the stated bounded-ratio assumptions and the strict quadratic behavior given by $f''(1)>0$. Adding a term proportional to $u-1$ to $f(u)$ does not change~\eqref{eq:f-divergence}, since $\int q(p/q-1)\dd\mu=\int p\dd\mu-\int q\dd\mu=0$; hence, after replacing $f(u)$ by $f(u)-f'(1)(u-1)$, we may assume $f'(1)=0$ in the local argument. Under the local bounded-ratio assumptions used below, the argument applies to the standard $f$-divergence generators for KL, reverse KL, Hellinger-type divergences, Pearson and Neyman chi-square divergences, and Jensen--Shannon divergence~\cite{kullback1951information,ali1966general,csiszar1967information,lin1991divergence}.

\begin{theorem}
\label{thm:f-local-square-root}
Let $f$ be twice continuously differentiable in a neighborhood of $1$, with
$f(1)=f'(1)=0$ and $f''(1)>0$. Let $p$ be a PDF that is positive almost everywhere, and let $h$ be measurable with
$\int h\dd\mu=0$, $\int h^2/p\dd\mu\in(0,\infty)$, and $h/p$
essentially bounded. For sufficiently small $\varepsilon>0$, let
\[q_\varepsilon=p+\varepsilon h
  \quad\text{and}\quad
  p_{\varepsilon,t}=[p,q_\varepsilon]_t=p+t\varepsilon h .
\]
For fixed $a,b>0$, the balance equation
\begin{equation}
  aD_f(p\Vert p_{\varepsilon,t})
  =
  bD_f(q_\varepsilon\Vert p_{\varepsilon,t})
  \label{eq:f-balance-local}
\end{equation}
has a solution $t\in(0,1)$ for all sufficiently small $\varepsilon$. If $t_\varepsilon\in[0,1]$ is any selected solution, then
\begin{equation}
  t_\varepsilon\longrightarrow
  \frac{\sqrt b}{\sqrt a+\sqrt b}
  \qquad (\varepsilon\downarrow0).
  \label{eq:sqrt-local-limit}
\end{equation}
\end{theorem}

\begin{proof}
Let $u=h/p$ and $I=\int u^2p\dd\mu\in(0,\infty)$. The bounded-ratio assumptions keep all likelihood ratios entering $f$ in a compact neighborhood of $1$ for small $\varepsilon$, so the endpoint maps are continuous. Since $I>0$ and $f''(1)>0$, the endpoint divergences are positive for all sufficiently small $\varepsilon$, and the balance function changes sign; hence a balancing root exists by the intermediate value theorem. Uniform Taylor expansion at $1$, with remainders uniform in $t\in[0,1]$, gives
\begin{equation*}
D_f(p\Vert p_{\varepsilon,t})=\frac{f''(1)}{2}t^2\varepsilon^2I+o(\varepsilon^2)\quad\text{and}\quad
D_f(q_\varepsilon\Vert p_{\varepsilon,t})=\frac{f''(1)}{2}(1-t)^2\varepsilon^2I+o(\varepsilon^2).
\end{equation*}
Dividing the balance equation by $(f''(1)/2)\varepsilon^2I$ and passing to subsequential limits gives $a\tau^2=b(1-\tau)^2$, whose unique solution in $[0,1]$ is~\eqref{eq:sqrt-local-limit}.
\end{proof}

Theorem~\ref{thm:f-local-square-root} gives the endpoint-to-candidate divergence barrier. Norm-induced balance gives $b/(a+b)$ because endpoint distance grows linearly along the segment. Smooth $f$-divergence balance gives $\sqrt b/(\sqrt a+\sqrt b)$ locally because the divergence grows quadratically. Thus, all such divergences agree  with the transformed-weight monoid of Proposition~\ref{prop:effective-weight} for $g(w)=\sqrt w$, not with additive supplied weights.

The same observation yields a positive first-order result when the propagated
mass follows the square-root law. For a binary aggregation tree $T$ with leaves
labeled by weighted states $(w_i,p_i^\varepsilon)$, we use the following
square-root propagated $D_f$-balancing protocol. At each internal node, let
$(W_A,r_A^\varepsilon)$ and $(W_B,r_B^\varepsilon)$ be the two child states.
Choose a root of the endpoint-to-candidate $D_f$-balance equation for
$r_A^\varepsilon$ and $r_B^\varepsilon$. The output density is the
corresponding segment point, and the output weight is
\[ W_{A\cup B}=(\sqrt{W_A}+\sqrt{W_B})^2 .
\]

\begin{corollary}
\label{cor:sqrt-first-order}
Let $f$ satisfy the assumptions of Theorem~\ref{thm:f-local-square-root}, and
let $p$ be a PDF that is positive almost everywhere. For $i=1,\ldots,m$, let
$p_i^\varepsilon=p+\varepsilon h_i$, where $\int h_i\dd\mu=0$, $h_i/p$
is essentially bounded, and $p_i^\varepsilon$ is a PDF for all sufficiently
small $\varepsilon>0$. Let $w_1,\ldots,w_m>0$.

For every binary aggregation tree $T$ with leaves labeled by
$(w_i,p_i^\varepsilon)$, the square-root propagated $D_f$-balancing protocol
has roots at all internal nodes for sufficiently small $\varepsilon$, and its
output density satisfies
\begin{equation}
  p^\varepsilon_{\mathrm{out}}
  =
  p+\varepsilon\,
  \frac{\sum_{i=1}^m \sqrt{w_i}\,h_i}{\sum_{i=1}^m\sqrt{w_i}}
  +o(\varepsilon),
  \label{eq:first-order-square-root}
\end{equation}
where the remainder is understood in the finite-dimensional perturbation space
$\operatorname{span}\{h_1,\ldots,h_m\}$ with norm
$\|r\|_{p,\infty}=\|r/p\|_\infty$. Thus the first-order term is
independent of the aggregation tree.
\end{corollary}

\begin{proof}
For a subtree with leaf set $A$, let
$G_A=\sum_{i\in A}\sqrt{w_i}$. We specify $H_A$ by
$G_AH_A=\sum_{i\in A}\sqrt{w_i}h_i$. Induction on the tree shows that
the subtree has propagated weight $G_A^2$ and density
$p+\varepsilon H_A+o(\varepsilon)$. All $o(\varepsilon)$ terms are taken in
the finite-dimensional span generated by the $h_i$; equivalently, after
division by $p$, the corresponding remainders are uniformly bounded and
converge to zero in that span. Since intermediate densities are convex
combinations of the leaf densities, the bounded-ratio assumptions persist, so
the local balance equation has a root for sufficiently small $\varepsilon$. At
a parent with child leaf sets $A,B$, the two-perturbation Taylor expansion of
the local balance equation gives segment coefficient
$G_B/(G_A+G_B)+o(1)$ unless $H_A=H_B$, in which case every coefficient gives
the same first-order term. The parent density is therefore
$p+\varepsilon(G_AH_A+G_BH_B)/(G_A+G_B)+o(\varepsilon)
=p+\varepsilon H_{A\cup B}+o(\varepsilon)$, and the propagated weight is
$(G_A+G_B)^2=G_{A\cup B}^2$. At the root this gives
\eqref{eq:first-order-square-root}, independently of the aggregation tree.
\end{proof}

Corollary~\ref{cor:sqrt-first-order} does not restore order-invariance for
additive source weights. It identifies a first-order local repair: under square-root propagated weights, smooth $f$-divergence segment balancing is tree-independent to first order in the local perturbation parameter, and exact associativity still requires propagating the transformed effective weight required by Proposition~\ref{prop:effective-weight}. The next result
shows that this is not a criticism of divergence objectives in general. For the forward objective considered below, global barycenter objectives have a different local limit.

\begin{proposition}
\label{prop:global-barycenter-local}
Let $f$ and $p$ satisfy the assumptions of
Theorem~\ref{thm:f-local-square-root}. For $i=1,\ldots,m$, let
$p_i^\varepsilon=p+\varepsilon h_i$, where $\int h_i\dd\mu=0$, $h_i/p$
is essentially bounded, and $\int h_i^2/p\dd\mu<\infty$. Let
$w_1,\ldots,w_m>0$. Let
$q_k^\varepsilon=p+\varepsilon k$, where $k$ has the same zero-integral,
bounded-ratio, and square-integrability properties, and assume
$q_k^\varepsilon$ is positive for sufficiently small $\varepsilon$. Then,
\begin{equation}
  \sum_{i=1}^m w_iD_f(p_i^\varepsilon\Vert q_k^\varepsilon)
    = \frac{f''(1)}{2}\varepsilon^2
  \sum_{i=1}^m w_i\int_\Omega\frac{(h_i-k)^2}{p}\dd\mu
  +o(\varepsilon^2).
  \label{eq:global-bary-local}
\end{equation}
Consequently, within the local linear perturbation space, the quadratic limiting objective is minimized by
\begin{equation}
  \bar h=\frac{\sum_{i=1}^m w_i h_i}{\sum_{i=1}^m w_i}.
  \label{eq:additive-bary-perturbation}
\end{equation}
\end{proposition}

\begin{proof}
For $u_i=h_i/p$ and $v=k/p$, the bounded-ratio assumptions give
\begin{equation*}
\frac{p_i^\varepsilon}{q_k^\varepsilon}=\frac{1+\varepsilon u_i}{1+\varepsilon v}=1+\varepsilon(u_i-v)+O(\varepsilon^2).
\end{equation*}
Uniform Taylor expansion at $1$ yields~\eqref{eq:global-bary-local}. The limiting quadratic objective equals a constant plus $(\sum_iw_i)\int (k-\bar h)^2/p\dd\mu$, because $\sum_iw_i(h_i-\bar h)=0$; hence $\bar h$ minimizes it within the local linear perturbation space.
\end{proof}

Proposition~\ref{prop:global-barycenter-local} separates two uses of divergences. Global $f$-barycenter objectives have the usual additive-weight target in their local quadratic limit; the proposition does not assert exact equality to linear pooling away from the local model. The square-root law is specific to endpoint-to-candidate binary balancing. With additive output weights, the following incompatibility remains.

\begin{corollary}
\label{cor:f-incompatibility}
No endpoint-to-candidate segment-balancing rule generated by a smooth $f$-divergence as in Theorem~\ref{thm:f-local-square-root} can, on all sufficiently close endpoint pairs satisfying the stated local assumptions, simultaneously satisfy all of the following requirements: additive output weight $a+b$, continuous coefficient depending only on $a$ and $b$, commutativity, associativity over all density classes with three affinely independent densities, and divergence balance.
\end{corollary}

\begin{proof}
Theorem~\ref{thm:converse} forces coefficient $b/(a+b)$ under the algebraic assumptions. For $a\neq b$, this value differs from $\sqrt b/(\sqrt a+\sqrt b)$. Theorem~\ref{thm:f-local-square-root} supplies arbitrarily close endpoint pairs whose selected $f$-divergence balancing coefficients converge to the latter value, so the listed requirements cannot all hold.
\end{proof}

KL balancing gives a finite witness of the same obstruction. For the numerical cases below, let $\oplus_{\KL}$ denote the binary rule that assigns output weight $w_a+w_b$ and chooses the unique segment point satisfying the KL balancing equation. Let $p_x$ be the Bernoulli density assigning probability $x$ to outcome $1$. For endpoints $x<y$, a segment point has Bernoulli parameter $z=(1-t)x+ty$, equivalently $t=(z-x)/(y-x)$, and the KL balancing equation is
\begin{equation*}
  w_a\KL(p_x\Vert p_z)=w_b\KL(p_y\Vert p_z).
\end{equation*}
For Bernoulli parameters $x<y$, the derivative of $w_a\KL(p_x\Vert p_z)-w_b\KL(p_y\Vert p_z)$ with respect to $z$ is
\begin{equation*}
  \frac{(w_a-w_b)z-w_ax+w_by}{z(1-z)}.
\end{equation*}
\begin{example}\label{ex:kl-nonassociative}
Consider the three equal-weight Bernoulli inputs $(1,p_{0.1}), (1,p_{0.6})$ and $(1,p_{0.9})$.
For the balancing equations used in this example, the derivative above is
strictly positive on the relevant intervals, so each root is unique.
One-dimensional root solving gives
\[(1,p_{0.1})\oplus_{\KL}(1,p_{0.6})=(2,p_{0.332731})\,
  \,\text{and}\,\,
  (1,p_{0.6})\oplus_{\KL}(1,p_{0.9})=(2,p_{0.761290}).\]
The two parenthesizations then give $((1,p_{0.1})\oplus_{\KL}(1,p_{0.6}))\oplus_{\KL}(1,p_{0.9})
  =(3,p_{0.581166})$, whereas $(1,p_{0.1})\oplus_{\KL}((1,p_{0.6})\oplus_{\KL}(1,p_{0.9}))
  =(3,p_{0.478333})$. The final Bernoulli parameters differ by approximately $0.102834$, so repeated local KL balancing is not associative.
\end{example}

\section{Compression Congruences for Gaussian Mixtures}\label{sec:gmm}
We now instantiate the algebraic results in a finite model class. Gaussian mixtures are closed under exact weighted linear fusion, but exact closure may
increase the number of mixture components. Stepwise compression preserves
order-invariance only under an algebraic compatibility condition.

A finite Gaussian-mixture density on $\R^d$ has the form
\begin{equation}
  p(x)=\sum_{r=1}^m \alpha_r\N(x;\mu_r,\Sigma_r),
  \qquad
  \alpha_r\geq0
  \quad\text{and}\quad
  \sum_{r=1}^m\alpha_r=1,
  \label{eq:gmm}
\end{equation}
where each covariance matrix $\Sigma_r$ is positive definite. A component is one Gaussian term in the chosen finite representation. Components with zero coefficient may be removed, and identical components may be merged.

The closure property is elementary but important: exact weighted fusion is
linear pooling at the density level.

\begin{proposition}
\label{prop:gmm-exact-closure}
For $i=1,\ldots,n$, let $p_i(x)=\sum_{r=1}^{m_i}\alpha_{ir}\N(x;\mu_{ir},\Sigma_{ir})$ be a finite Gaussian mixture and let $w_i>0$. The exact weighted fusion
\begin{equation}
  p_{\mathrm{fused}}(x)
  =
  \sum_i\sum_{r=1}^{m_i}
  \frac{w_i\alpha_{ir}}{\sum_jw_j}
  \N(x;\mu_{ir},\Sigma_{ir})
  \label{eq:gmm-fusion}
\end{equation}
 is a Gaussian mixture with at most $\sum_i m_i$ component terms. If all component weights in the concatenated representation are positive and no two components have the same mean and covariance, that representation has exactly $\sum_i m_i$ nonzero components.
\end{proposition}
\begin{proof}
Substituting each mixture expansion into $\sum_i w_i p_i/\sum_j w_j$
gives~\eqref{eq:gmm-fusion}. The coefficients in~\eqref{eq:gmm-fusion} are nonnegative and sum to one. Duplicate or zero-weight component terms may be removed. If all component weights are positive and no duplicates occur, the concatenated representation has exactly $\sum_i m_i$ nonzero components.
\end{proof}

Proposition~\ref{prop:gmm-exact-closure} describes a density-level aggregate of already trained local models. It does not infer latent component correspondences, merge nearby components, or optimize likelihood under a prescribed component budget. For example, exact post-hoc fusion of two one-dimensional Gaussians $\N(-1,\sigma^2)$ and $\N(1,\sigma^2)$ with equal weights gives a two-component mixture, whereas the population maximum-likelihood one-Gaussian approximation has mean $0$ and variance $\sigma^2+1$~\cite{dempster1977maximum,mclachlan2000finite}.

To study compression, associate with a weighted Gaussian-mixture state $(w,p)$, where $p=\sum_r\alpha_r\N(\cdot;\mu_r,\Sigma_r)$, the unnormalized component measure $M_{(w,p)}=\sum_r\beta_r\delta_{(\mu_r,\Sigma_r)}$ with $\beta_r=w\alpha_r$. Exact weighted fusion corresponds to addition of these unnormalized measures followed by normalization by total mass. We therefore use unnormalized component measures 
\begin{equation}
  M=\sum_{r=1}^m \beta_r\delta_{\theta_r},
  \qquad
  \theta_r=(\mu_r,\Sigma_r)
  \quad\text{and}\quad
  \beta_r\geq0,
  \label{eq:component-measure}
\end{equation}
including the zero measure. Here $\delta_\theta$ denotes the Dirac measure at $\theta$, and $W(M)=\sum_r\beta_r$ denotes the total mass. This unnormalized representation is different from the formal identity $\e$ for positive PDF states: here $0$ is an actual zero component measure. Exact fusion before final normalization is addition in the commutative monoid $\mathcal M$ of such measures. Let $C:\mathcal M\to\mathcal M$ be an idempotent compression map with $C(0)=0$. When $C$ is used to compress weighted PDF states, we additionally require total-mass preservation $W(C(M))=W(M)$ unless an explicit weight renormalization is intended. Stepwise compressed fusion is
\begin{equation}
  M\odot_C N=C(M+N).
  \label{eq:compressed-product}
\end{equation}
The question is when evaluating~\eqref{eq:compressed-product} at every internal node gives the same result as exact summation followed by one final compression.

\begin{theorem}\label{thm:compression-congruence}
Let $C:\mathcal M\to\mathcal M$ be idempotent with $C(0)=0$. The following
conditions are equivalent.
\begin{enumerate}
  \item For every finite family $M_1,\ldots,M_n\in\mathcal M$, $n\geq2$,
  every binary-tree evaluation using $\odot_C$ returns
  $C(\sum_i M_i)$.

  \item $C$ is incrementally compatible with addition:
  \begin{equation}
    C(C(M)+N)=C(M+N)
    \label{eq:incremental-compatibility}
  \end{equation}
  for all $M,N\in\mathcal M$.

  \item The equivalence relation $M\sim N$ iff $C(M)=C(N)$ is a congruence
  for addition: $M\sim N$ implies $M+R\sim N+R$ for every $R\in\mathcal M$.
\end{enumerate}
\end{theorem}

\begin{proof}
Assume~\eqref{eq:incremental-compatibility}. Induction over binary trees shows that every nontrivial subtree with leaf set $I$ evaluates to $C(\sum_{i\in I}M_i)$; the induction step uses~\eqref{eq:incremental-compatibility} once or twice depending on whether the children are leaves. Conversely, the tree condition applied to the allowed three-leaf family $M,0,N$ gives $C(C(M)+N)=C(M+N)$. Finally, idempotence gives $M\sim C(M)$, so congruence implies incremental compatibility; the reverse follows from $C(M+R)=C(C(M)+R)=C(C(N)+R)=C(N+R)$ whenever $C(M)=C(N)$.
\end{proof}

The proof uses only the commutative-monoid laws, so the equivalence holds with addition replaced by any commutative monoid operation. For an increasing bijection $g:\R_{>0}\to\R_{>0}$, let $\psi_g(0)=0$ and $\psi_g(M)=g(W(M))M/W(M)$ for $W(M)>0$. The map is bijective, and the inverse rescales every nonzero $M$ by $g^{-1}(W(M))/W(M)$. Then, $M\boxplus_gN=\psi_g^{-1}(\psi_g(M)+\psi_g(N))$; Proposition~\ref{prop:effective-weight} gives the corresponding scalar construction.

The theorem reduces compression-compatible fusion to a monoid-congruence
condition. Identity compression satisfies the condition. Moment compression to one Gaussian also satisfies it for $W(M)>0$, with $C(0)=0$ by convention, when the state stores $W(M)$, the first moment $m(M)=\sum_r\beta_r\mu_r$, and the second raw moment $S(M)=\sum_r\beta_r(\Sigma_r+\mu_r\mu_r^\top)$; these summaries are additive and preserved by compression. Since the component covariances are positive definite and $W(M)>0$, the moment-matched covariance is positive definite. Fixed-dictionary compression satisfies the condition for the same reason provided the dictionary assignment is fixed, componentwise, independent of the current batch or aggregation tree, and assigns dictionary atoms to themselves so that $C$ is idempotent: it stores additive mass assigned to each dictionary element.

For weighted Gaussian states, let $\oplus_C$ mean exact weighted fusion followed by the compression $C$ on the associated unnormalized component measure. Common local reduction heuristics require more care. Gaussian-mixture reduction by pruning or pairwise merging is a standard approximation problem~\cite{runnalls2007kullback}, but such reduction does not automatically define a schedule-independent fusion rule.

\begin{example}\label{ex:gmm-pruning-nonassociative}
Let $C$ retain the Gaussian component with largest normalized
coefficient, equivalently largest unnormalized component mass within a fixed
total weight. The total state weight is preserved, and ties are broken by choosing
the smaller mean. Consider
\[x=(1,\N(0,1)),\qquad
  y=(1,\N(1,1))\quad\text{and}\quad
  z=(1,\N(2,1)).
\]
Then $C(x\oplus y)=(2,\N(0,1))$ and $ C(y\oplus z)=(2,\N(1,1))$.
Therefore,
\[ (x\oplus_C y)\oplus_C z=(3,\N(0,1))
  \quad\text{and}\quad
  x\oplus_C(y\oplus_C z)=(3,\N(1,1)).
\]
The two final states differ, so local pruning is not associative.
\end{example}

Lossy compression is therefore not an implementation detail unless it satisfies
the congruence condition or carries an explicit approximation guarantee.
Example~\ref{ex:gmm-pruning-nonassociative} separates exact semantic protocols
from local reduction heuristics: exact weighted fusion and additive summaries are
tree-independent, while local pruning is not.

\section{Related Work and Discussion}\label{sec:related} 
Linear and logarithmic opinion pools are standard rules for aggregating
probability distributions and expert forecasts
\cite{stone1961opinion,genest1986combining,jacobs1995methods}. Our question is
not which pooling formula to choose, but when a binary pooling rule can be used
as a schedule-independent distributed protocol. Within continuous
segment-valued rules with additive weights and weight-only coefficients,
Theorem~\ref{thm:converse} shows that order-invariant iteration forces
normalized linear pooling. This places Theorem~\ref{thm:converse} close to representation results for associative means and functional equations~\cite{aczel2006lectures}, but our setting is weighted density states and hierarchical execution.

Dempster--Shafer theory, the transferable belief model, and distances between bodies of evidence provide broader frameworks for uncertain evidence fusion~\cite{dempster2008upper,shafer1976mathematical,smets1994transferable,jousselme2001new}. Distributional averages based on $f$-divergences, Bregman divergences, or Wasserstein barycenters define global aggregation principles~\cite{ali1966general,csiszar1967information,banerjee2005clustering,nielsen2009sided,agueh2011barycenters}. These principles answer a different question: a global barycenter may be well posed, but this does not imply that repeatedly solving binary barycenter or balancing subproblems along an arbitrary tree is associative.

Den{\oe}ux studies distributed combination when agents at the nodes of a communication network hold belief functions~\cite{denoeux2021distributed}; our analysis instead concerns tree-independent binary fusion of weighted densities.

Gaussian-mixture learning and Gaussian-mixture reduction are adjacent but distinct~\cite{dempster1977maximum,mclachlan2000finite,runnalls2007kullback,zhang2022distributed}. Distributed learning of Gaussian mixtures estimates mixture parameters from partitioned data or sufficient statistics. Mixture-reduction algorithms compress a given mixture under an approximation criterion. We instead assume already trained local densities and ask which exact or compressed post-hoc fusion procedures have a tree-independent semantics. The divergence and compression results are compositional boundary results: pairwise solvability and model-size reduction are not enough for schedule-independent uncertainty fusion. Our characterizations apply to distributed sensor networks, federated density ensembles, and hierarchical mixture reduction, where communication and scheduling determine the aggregation tree; they distinguish rules that tolerate arbitrary tree changes from rules requiring transformed weights or congruence-compatible summaries.

\section{Conclusion}\label{sec:conclusion}

We have studied weighted density fusion as a compositionality problem for scalable uncertainty management. Our results give exact algebraic conditions for schedule-independent binary fusion in the segment-valued additive class. In this class, norm-induced segment balancing yields weighted linear pooling, and commutative, associative hierarchical execution forces the normalized-weight coefficient. Transformed-weight rules avoid this conclusion only by changing the quantity propagated additively by the state.

Smooth endpoint-to-candidate $f$-divergence balancing does not satisfy the
same additive-weight compositional requirement. Under the bounded-ratio local
model of Theorem~\ref{thm:f-local-square-root}, its quadratic geometry induces
square-root effective weights rather than additive source weights, yielding only
first-order local tree-independence under square-root propagation; the Bernoulli
KL example gives a finite associativity failure. Global $f$-barycenter
objectives behave differently: their local quadratic limit retains additive-weight
targets. For Gaussian mixtures, exact fusion is schedule-independent but may
increase the number of mixture components, while stepwise compression is order-independent
exactly when idempotent compression is compatible with unnormalized component
addition. Future work should quantify how endpoint-dependent divergence rules and
approximate compression maps deviate from associativity, for example through
explicit associator or tree-dependent error bounds.

\paragraph{\textbf{Acknowledgments.}}
We acknowledge funding by the German Research Foundation (DFG), SFB 1574 ``Circular Factory,'' project number 471687386.

\bibliographystyle{splncs04}
\bibliography{ref}

\clearpage
\appendix
\renewcommand{\theHsection}{appendix.\Alph{section}}
\section{Extended Proofs}\label{app:proofs}

Notations are the same as in Sections~\ref{sec:balance}--\ref{sec:gmm}.

\subsection*{Proof of Theorem~\ref{thm:norm-characterization}}

Let $p_t=[p_a,p_b]_t=(1-t)p_a+tp_b$. Then,
\begin{equation*}
  p_t-p_a=t(p_b-p_a)
  \quad\text{and}\quad
  p_t-p_b=(1-t)(p_a-p_b).
\end{equation*}
Since $\Delta(p,q)=\norm{p-q}$ is induced by a norm, absolute homogeneity gives
\begin{equation*}
  \Delta(p_a,p_t)=\norm{p_a-p_t}=t\norm{p_b-p_a}
\end{equation*}
and
\begin{equation*}
  \Delta(p_b,p_t)=\norm{p_b-p_t}=(1-t)\norm{p_b-p_a}.
\end{equation*}
If $p_a\neq p_b$ in the almost-everywhere quotient, then $\norm{p_b-p_a}>0$. The balance equation is therefore equivalent to
\begin{equation*}
  w_a t\norm{p_b-p_a}=w_b(1-t)\norm{p_b-p_a},
\end{equation*}
and cancellation of the positive common factor yields $w_at=w_b(1-t)$. Solving gives
\begin{equation*}
  t^*=\frac{w_b}{w_a+w_b}.
\end{equation*}
Substitution into $[p_a,p_b]_t$ gives the stated fused density. The scalar equation in $t$ is linear with nonzero coefficient $w_a+w_b$, so the balancing point is unique. If $p_a=p_b$ almost everywhere, then $[p_a,p_b]_t=p_a$ for all $t\in[0,1]$. The parameter may not be unique, but the density is unique.

\subsection*{Proof of Proposition~\ref{prop:monoid}}

For positive states, the operation in~\eqref{eq:oplus} gives
\begin{equation*}
  (w_a,p_a)\oplus(w_b,p_b)
  =
  \left(w_a+w_b,\frac{w_ap_a+w_bp_b}{w_a+w_b}\right).
\end{equation*}
Applying $\eta$ gives
\begin{align*}
  \eta((w_a,p_a)\oplus(w_b,p_b))
  &=\left(w_a+w_b,(w_a+w_b)\frac{w_ap_a+w_bp_b}{w_a+w_b}\right)\\
  &=(w_a+w_b,w_ap_a+w_bp_b)\\
  &=(w_a,w_ap_a)+(w_b,w_bp_b)\\
  &=\eta(w_a,p_a)+\eta(w_b,p_b).
\end{align*}
The element $\e$ satisfies $\e\oplus x=x\oplus\e=x$ by definition, and $\eta(\e)=(0,0)$ is the additive identity in represented coordinates. Closure follows from convexity of $\Pcal$: the second coordinate of the fused state is a convex combination of $p_a$ and $p_b$, because $w_a,w_b>0$ and the coefficients sum to one. Since addition of represented pairs is commutative and associative, the original operation is commutative and associative as well. For a finite nonempty family, repeated use of the additive representation yields
\begin{equation*}
  \eta(x_1\oplus\cdots\oplus x_m)
  =\sum_{i=1}^m\eta(w_i,p_i)
  =\left(\sum_{i=1}^m w_i,\sum_{i=1}^m w_ip_i\right).
\end{equation*}
Recovering the weighted state from the represented pair gives~\eqref{eq:finite-fusion}. Commutativity and associativity make the result independent of the parenthesization and order.

\subsection*{Proof of Lemma~\ref{lem:cocycle-representation}}

The assumptions give a continuous function $m:\R_{>0}^2\to(0,1)$ satisfying
\begin{equation*}
  m(a,b)m(a+b,c)=m(a,b+c).
\end{equation*}
Since $m$ is positive, the function $h(a,b)=-\log m(a,b)$ is continuous and satisfies
\begin{equation}
  h(a,b)+h(a+b,c)=h(a,b+c)
  \label{eq:app-cocycle-h}
\end{equation}
for all $a,b,c>0$. We construct a potential $H$ with $h(a,b)=H(a+b)-H(a)$.
Choose $r>0$. We first justify the boundary value at $r$. For fixed $a,c>0$, equation~\eqref{eq:app-cocycle-h} gives $h(a,b)+h(a+b,c)=h(a,b+c)$ for every $b>0$. Letting $b\downarrow0$ and using continuity of $h$ on the positive domain gives $\lim_{b\downarrow0}h(a,b)=0$. Hence $h(r,x-r)\to0$ as $x\downarrow r$ and $h(x,r-x)\to0$ as $x\uparrow r$. On $(r,\infty)$, set $H(x)=h(r,x-r)$, set $H(r)=0$, and on $(0,r)$ set $H(x)=-h(x,r-x)$. The preceding limit shows that $H$ is continuous at $r$.

We verify the representation. If $a>r$, then $a+b>r$, and~\eqref{eq:app-cocycle-h} with first argument $r$, second argument $a-r$, and third argument $b$ gives
\begin{equation*}
  h(r,a-r)+h(a,b)=h(r,a+b-r).
\end{equation*}
Thus,
\begin{equation*}
  h(a,b)=H(a+b)-H(a).
\end{equation*}
If $a+b<r$, then $a<r$. Applying~\eqref{eq:app-cocycle-h} with first argument $a$, second argument $b$, and third argument $r-a-b$ gives
\begin{equation*}
  h(a,b)+h(a+b,r-a-b)=h(a,r-a),
\end{equation*}
so
\begin{equation*}
  h(a,b)=-h(a+b,r-a-b)+h(a,r-a)=H(a+b)-H(a).
\end{equation*}
If $a<r<a+b$, let $c=r-a>0$ and $d=a+b-r>0$. Applying~\eqref{eq:app-cocycle-h} to $(a,c,d)$ gives
\begin{equation*}
  h(a,r-a)+h(r,a+b-r)=h(a,b).
\end{equation*}
By the definition of $H$, this is again $h(a,b)=H(a+b)-H(a)$. The boundary cases follow by continuity. Hence the representation holds for all positive $a,b$.

Let $F(x)=\exp H(x)$. Then, $F$ is positive and continuous, and
\begin{equation*}
  m(a,b)=\exp(-h(a,b))=\frac{F(a)}{F(a+b)}.
\end{equation*}
If another positive continuous $\widetilde F$ gives the same representation, then $F(a)/\widetilde F(a)=F(a+b)/\widetilde F(a+b)$ for all $a,b>0$. For any $x<y$, choose $b=y-x$ and obtain the same ratio at $x$ and $y$; hence the ratio is constant on $\R_{>0}$. Thus, $F$ is unique up to multiplication by a positive constant.

\subsection*{Proof of Theorem~\ref{thm:converse}}

Let $m(a,b)=1-\alpha(a,b)$. Then $m$ is continuous and takes values in $(0,1)$. For positive weights $a,b$, the operation has the form
\begin{equation*}
  (a,p)\odot(b,q)=(a+b,m(a,b)p+(1-m(a,b))q).
\end{equation*}
Choose a convex PDF class containing three affinely independent densities $p_a,p_b,p_c$. Associativity says that the density component of
\begin{equation*}
  ((a,p_a)\odot(b,p_b))\odot(c,p_c)
\end{equation*}
coincides with the density component of
\begin{equation*}
  (a,p_a)\odot((b,p_b)\odot(c,p_c)).
\end{equation*}
Expanding the left parenthesization gives the coefficient of $p_a$ as $m(a,b)m(a+b,c)$. Expanding the right parenthesization gives the coefficient of $p_a$ as $m(a,b+c)$. Since $p_a,p_b,p_c$ are affinely independent, equality of the two convex combinations implies equality of the corresponding barycentric coefficients. Therefore,
\begin{equation*}
  m(a,b)m(a+b,c)=m(a,b+c)
\end{equation*}
for all $a,b,c>0$.

Lemma~\ref{lem:cocycle-representation} gives a positive continuous function $F$ with
\begin{equation*}
  m(a,b)=\frac{F(a)}{F(a+b)}.
\end{equation*}
Commutativity of $\odot$ means that the density produced from $(a,p)$ and $(b,q)$ is the same as the density produced from $(b,q)$ and $(a,p)$. Comparing the coefficients of $p$ and $q$ gives
\begin{equation*}
  m(a,b)=1-m(b,a),
\end{equation*}
or equivalently $m(a,b)+m(b,a)=1$. Substituting the representation gives
\begin{equation*}
  \frac{F(a)}{F(a+b)}+\frac{F(b)}{F(a+b)}=1.
\end{equation*}
Thus,
\begin{equation*}
  F(a+b)=F(a)+F(b)
\end{equation*}
for all positive $a,b$. Since $F$ is continuous and positive on $\R_{>0}$, the Cauchy equation on the positive reals has the form $F(w)=cw$ for some constant $c>0$. Hence
\begin{equation*}
  m(a,b)=\frac{a}{a+b}
  \quad\text{and}\quad
  \alpha(a,b)=1-m(a,b)=\frac{b}{a+b}.
\end{equation*}

\subsection*{Proof of Proposition~\ref{prop:effective-weight}}

Let
\begin{equation*}
  \eta_g(w,p)=(g(w),g(w)p)
  \quad\text{and}\quad
  \eta_g(\e)=(0,0).
\end{equation*}
For positive states, the operation in~\eqref{eq:effective-weight-rule} gives
\begin{align*}
  \eta_g((a,p)\oplus_g(b,q))
  &=\left(g(a)+g(b), (g(a)+g(b))\frac{g(a)p+g(b)q}{g(a)+g(b)}\right)\\
  &=(g(a)+g(b),g(a)p+g(b)q)\\
  &=\eta_g(a,p)+\eta_g(b,q).
\end{align*}
The first represented coordinate is positive for every positive state, and the original state is recovered from $(s,r)$ with $s>0$ as $(g^{-1}(s),r/s)$. The formal element $\e$ corresponds to the additive zero. Therefore, addition of represented pairs transfers closure, commutativity, associativity, and the identity laws to $\oplus_g$. Iteration gives
\begin{equation*}
  \eta_g(x_1\oplus_g\cdots\oplus_g x_m)
  =\left(\sum_{i=1}^m g(w_i),\sum_{i=1}^m g(w_i)p_i\right),
\end{equation*}
which recovers~\eqref{eq:effective-weight-finite}.

\subsection*{Proof of Theorem~\ref{thm:f-local-square-root}}

Let $u=h/p$. The assumptions give $\int u p\dd\mu=0$, $I=\int u^2p\dd\mu\in(0,\infty)$, and $\norm{u}_{\infty}<\infty$. For sufficiently small $\varepsilon>0$, the quantities $1+\varepsilon u$ and $1+t\varepsilon u$ are bounded away from zero uniformly in $t\in[0,1]$. Hence $q_\varepsilon=p(1+\varepsilon u)$ and $p_{\varepsilon,t}=p(1+t\varepsilon u)$ are nonnegative densities, and all ratios entering $f$ lie in a compact interval contained in the neighborhood where $f$ is twice continuously differentiable. The endpoint maps are finite and continuous; the expansions below show that the balance function is negative at $t=0$ and positive at $t=1$ for sufficiently small $\varepsilon$, so at least one root exists.

For the left endpoint divergence, the ratio is
\begin{equation*}
  \frac{p}{p_{\varepsilon,t}}=\frac{1}{1+t\varepsilon u}.
\end{equation*}
Uniformly in $t\in[0,1]$,
\begin{equation*}
  \frac{1}{1+t\varepsilon u}=1-t\varepsilon u+t^2\varepsilon^2u^2+O(\varepsilon^3),
\end{equation*}
where the remainder is dominated by a constant multiple of $\varepsilon^3$ because $u$ is essentially bounded. Since $f(1)=f'(1)=0$,
\begin{equation*}
  f(1+s)=\frac{f''(1)}{2}s^2+o(s^2)
\end{equation*}
uniformly for $s$ in a sufficiently small compact neighborhood of zero. Multiplying by $p_{\varepsilon,t}=p(1+t\varepsilon u)$, integrating, and using bounded domination gives
\begin{align*}
  D_f(p\Vert p_{\varepsilon,t})
  &=\int p(1+t\varepsilon u)
      f\!\left(\frac{1}{1+t\varepsilon u}\right)\dd\mu\\
  &=\frac{f''(1)}{2}t^2\varepsilon^2\int u^2p\dd\mu+o(\varepsilon^2)\\
  &=\frac{f''(1)}{2}t^2\varepsilon^2 I+o(\varepsilon^2),
\end{align*}
with the $o(\varepsilon^2)$ term uniform over $t\in[0,1]$.

For the right endpoint divergence,
\begin{equation*}
  \frac{q_\varepsilon}{p_{\varepsilon,t}}
  =\frac{1+\varepsilon u}{1+t\varepsilon u}
  =1+(1-t)\varepsilon u+O(\varepsilon^2)
\end{equation*}
uniformly in $t$. The same Taylor argument gives
\begin{equation*}
  D_f(q_\varepsilon\Vert p_{\varepsilon,t})
  =\frac{f''(1)}{2}(1-t)^2\varepsilon^2 I+o(\varepsilon^2),
\end{equation*}
uniformly in $t\in[0,1]$.

Let $t_\varepsilon$ be any selected solution of~\eqref{eq:f-balance-local}. The interval $[0,1]$ is compact, so every sequence $\varepsilon_k\downarrow0$ has a subsequence along which $t_{\varepsilon_k}\to\tau\in[0,1]$. Along that subsequence, divide the balance equation by $(f''(1)/2)\varepsilon_k^2I$. The two uniform expansions yield
\begin{equation*}
  a\tau^2=b(1-\tau)^2.
\end{equation*}
Both sides are nonnegative. Taking square roots gives $\sqrt a\,\tau=\sqrt b\,(1-\tau)$, and therefore
\begin{equation*}
  \tau=\frac{\sqrt b}{\sqrt a+\sqrt b}.
\end{equation*}
Every convergent subsequence has the same limit, so the whole family $t_\varepsilon$ converges to the stated value.

\subsection*{Proof of Corollary~\ref{cor:sqrt-first-order}}

For a set $A$ of leaves, let
\begin{equation*}
  G_A=\sum_{i\in A}\sqrt{w_i}
  \quad\text{and}\quad
  H_A=\frac{1}{G_A}\sum_{i\in A}\sqrt{w_i}\,h_i.
\end{equation*}

Let $V=\operatorname{span}\{h_1,\ldots,h_m\}$ and equip $V$ with the norm
$\|r\|_{p,\infty}:=\|r/p\|_\infty$. This norm is finite on $V$ because each
$h_i/p$ is essentially bounded. We prove by induction on subtrees that the
state computed for every subtree with leaf set $A$ has propagated weight
$G_A^2$ and density
\begin{equation*}
  p+\varepsilon H_A+r_A^\varepsilon
  \quad\text{with}\quad
  \|r_A^\varepsilon\|_{p,\infty}=o(\varepsilon).
\end{equation*}
Every intermediate density is a convex combination of the leaf densities, hence
has the form $p+\varepsilon H$ with $H$ in the compact convex hull generated by
the $h_i$. Therefore, the bounded-ratio assumptions persist uniformly at internal
nodes. The assertion is immediate for a leaf, with $r_A^\varepsilon=0$.

Assume the assertion for two child subtrees with leaf sets $A$ and $B$. Their densities have the form
\begin{equation*}
  p+\varepsilon H_A^\varepsilon
  \quad\text{and}\quad
  p+\varepsilon H_B^\varepsilon,
\end{equation*}
where $H_A^\varepsilon\to H_A$ and $H_B^\varepsilon\to H_B$ in $\|\cdot\|_{p,\infty}$ within the finite-dimensional span generated by the $h_i$. Let $\theta_\varepsilon$ be a balancing coefficient at the parent. A candidate point on the segment is
\begin{equation*}
  p+\varepsilon\bigl((1-\theta_\varepsilon)H_A^\varepsilon+\theta_\varepsilon H_B^\varepsilon\bigr).
\end{equation*}
The same uniform Taylor expansion as in Theorem~\ref{thm:f-local-square-root}, applied over the compact bounded-ratio set containing the two perturbations $H_A^\varepsilon$ and $H_B^\varepsilon$, gives
\begin{align*}
D_f\bigl(p+\varepsilon H_A^\varepsilon\Vert
p+\varepsilon((1-\theta)H_A^\varepsilon+\theta H_B^\varepsilon)\bigr)
&=\frac{f''(1)}{2}\varepsilon^2\theta^2 J_\varepsilon+o(\varepsilon^2),\\
D_f\bigl(p+\varepsilon H_B^\varepsilon\Vert
p+\varepsilon((1-\theta)H_A^\varepsilon+\theta H_B^\varepsilon)\bigr)
&=\frac{f''(1)}{2}\varepsilon^2(1-\theta)^2 J_\varepsilon+o(\varepsilon^2),
\end{align*}
where
\begin{equation*}
  J_\varepsilon=\int_\Omega \frac{(H_B^\varepsilon-H_A^\varepsilon)^2}{p}\dd\mu.
\end{equation*}
If $H_A\neq H_B$, then $J_\varepsilon$ converges to a positive limit. Dividing the balance equation by $(f''(1)/2)\varepsilon^2J_\varepsilon$ gives
\begin{equation*}
  G_A^2\theta_\varepsilon^2=G_B^2(1-\theta_\varepsilon)^2+o(1),
\end{equation*}
so every accumulation point of $\theta_\varepsilon$ is $G_B/(G_A+G_B)$. Hence $\theta_\varepsilon=G_B/(G_A+G_B)+o(1)$. If $H_A=H_B$, then both endpoints have the same first-order perturbation; hence every segment coefficient gives the same first-order perturbation, and no coefficient limit is needed in the induction step.

In both cases the parent density is
\begin{align*}
  p+\varepsilon\frac{G_AH_A+G_BH_B}{G_A+G_B}+o(\varepsilon)
  &=p+\varepsilon\frac{\sum_{i\in A\cup B}\sqrt{w_i}h_i}{G_A+G_B}+o(\varepsilon)\\
  &=p+\varepsilon H_{A\cup B}+o(\varepsilon).
\end{align*}
The propagated-weight rule gives
\begin{equation*}
  (\sqrt{G_A^2}+\sqrt{G_B^2})^2=(G_A+G_B)^2=G_{A\cup B}^2.
\end{equation*}
The induction proves the root claim, and the root is independent of the tree because both $G_{\{1,\ldots,m\}}$ and $H_{\{1,\ldots,m\}}$ are determined only by the leaves.

\subsection*{Proof of Proposition~\ref{prop:global-barycenter-local}}

Let $u_i=h_i/p$ and $v=k/p$. The assumptions imply that all $u_i$ and $v$ are essentially bounded and square-integrable with respect to the measure $p\dd\mu$, and $k$ belongs to the local linear perturbation class specified in Proposition~\ref{prop:global-barycenter-local}. For sufficiently small $\varepsilon$, both $p_i^\varepsilon=p(1+\varepsilon u_i)$ and $q_k^\varepsilon=p(1+\varepsilon v)$ are densities, and the ratios remain in a compact neighborhood of $1$.

For each $i$,
\begin{equation*}
  \frac{p_i^\varepsilon}{q_k^\varepsilon}
  =\frac{1+\varepsilon u_i}{1+\varepsilon v}
  =1+\varepsilon(u_i-v)+O(\varepsilon^2),
\end{equation*}
where the remainder is dominated by a constant times $\varepsilon^2$ on the bounded local class. Taylor expansion at $1$ gives
\begin{align*}
  D_f(p_i^\varepsilon\Vert q_k^\varepsilon)
  &=\int p(1+\varepsilon v)
      f\!\left(\frac{1+\varepsilon u_i}{1+\varepsilon v}\right)\dd\mu\\
  &=\frac{f''(1)}{2}\varepsilon^2
      \int (u_i-v)^2p\dd\mu+o(\varepsilon^2)\\
  &=\frac{f''(1)}{2}\varepsilon^2
      \int_\Omega\frac{(h_i-k)^2}{p}\dd\mu+o(\varepsilon^2).
\end{align*}
Summing over $i$ with weights $w_i$ gives~\eqref{eq:global-bary-local}.

For the limiting objective
\begin{equation*}
  Q(k)=\sum_{i=1}^m w_i\int_\Omega\frac{(h_i-k)^2}{p}\dd\mu,
\end{equation*}
let $\bar h=(\sum_iw_ih_i)/(\sum_iw_i)$. Expanding around $\bar h$ gives
\begin{align*}
  Q(k)
  &=\sum_i w_i\int\frac{(h_i-\bar h+\bar h-k)^2}{p}\dd\mu\\
  &=\sum_i w_i\int\frac{(h_i-\bar h)^2}{p}\dd\mu
    +\left(\sum_iw_i\right)\int\frac{(k-\bar h)^2}{p}\dd\mu,
\end{align*}
where the cross term vanishes because $\sum_iw_i(h_i-\bar h)=0$. The second term is nonnegative and equals zero exactly at $k=\bar h$ within the local linear perturbation space. Hence $\bar h$ is the minimizer there.

\subsection*{Proof of Corollary~\ref{cor:f-incompatibility}}

Assume that a rule satisfies the algebraic requirements listed in the corollary. Theorem~\ref{thm:converse} applies, because the rule has additive output weight, a continuous coefficient depending only on the two weights, and commutative associative execution over all convex density classes containing three affinely independent densities. Therefore, the segment coefficient must be
\begin{equation*}
  \frac{b}{a+b}.
\end{equation*}
Choose $a\neq b$. The equality
\begin{equation*}
  \frac{b}{a+b}=\frac{\sqrt b}{\sqrt a+\sqrt b}
\end{equation*}
would imply $b(\sqrt a+\sqrt b)=(a+b)\sqrt b$, hence $b\sqrt a=a\sqrt b$, and therefore $\sqrt a=\sqrt b$, contradicting $a\neq b$. Theorem~\ref{thm:f-local-square-root} gives sufficiently close endpoint pairs for which every selected $f$-divergence balancing coefficient converges to the square-root value. For small enough $\varepsilon$, that coefficient cannot equal the normalized additive-weight coefficient required by Theorem~\ref{thm:converse}. Hence the listed requirements cannot all hold.

\subsection*{Verification of the KL witness}

For a Bernoulli density $p_x$ on $\{0,1\}$, the KL divergence to $p_z$ is
\begin{equation*}
  \KL(p_x\Vert p_z)=x\log\frac{x}{z}+(1-x)\log\frac{1-x}{1-z}.
\end{equation*}
For fixed endpoints $x<y$ and positive weights $w_a,w_b$, the scalar balance function is
\begin{equation*}
  F(z)=w_a\KL(p_x\Vert p_z)-w_b\KL(p_y\Vert p_z)
  \quad\text{for}\quad z\in(x,y).
\end{equation*}
Differentiating gives
\begin{align*}
  F'(z)
  &=w_a\left(-\frac{x}{z}+\frac{1-x}{1-z}\right)
    -w_b\left(-\frac{y}{z}+\frac{1-y}{1-z}\right)\\
  &=\frac{(w_a-w_b)z-w_ax+w_by}{z(1-z)}.
\end{align*}
On each numerical interval used in the main text, the numerator is positive, and the denominator is positive because $z\in(0,1)$. Thus, the root is unique. Solving $F(z)=0$ by one-dimensional bisection or Newton iteration gives
\begin{equation*}
  (1,p_{0.1})\oplus_{\KL}(1,p_{0.6})=(2,p_{0.332731})
\end{equation*}
and
\begin{equation*}
  (1,p_{0.6})\oplus_{\KL}(1,p_{0.9})=(2,p_{0.761290}).
\end{equation*}
The second fusion step then gives
\begin{equation*}
  ((1,p_{0.1})\oplus_{\KL}(1,p_{0.6}))\oplus_{\KL}(1,p_{0.9})=(3,p_{0.581166})
\end{equation*}
and
\begin{equation*}
  (1,p_{0.1})\oplus_{\KL}((1,p_{0.6})\oplus_{\KL}(1,p_{0.9}))=(3,p_{0.478333}).
\end{equation*}
The final parameters differ by approximately $0.102833$, so the induced binary operation is not associative.

\subsection*{Proof of Proposition~\ref{prop:gmm-exact-closure}}

For each input mixture,
\begin{equation*}
  p_i(x)=\sum_{r=1}^{m_i}\alpha_{ir}\N(x;\mu_{ir},\Sigma_{ir}),
  \qquad
  \alpha_{ir}\ge0
  \quad\text{and}\quad
  \sum_{r=1}^{m_i}\alpha_{ir}=1.
\end{equation*}
The exact weighted fusion is
\begin{align*}
  \frac{\sum_iw_ip_i(x)}{\sum_jw_j}
  &=\frac{1}{\sum_jw_j}\sum_iw_i\sum_{r=1}^{m_i}\alpha_{ir}\N(x;\mu_{ir},\Sigma_{ir})\\
  &=\sum_i\sum_{r=1}^{m_i}\frac{w_i\alpha_{ir}}{\sum_jw_j}\N(x;\mu_{ir},\Sigma_{ir}).
\end{align*}
All coefficients in this representation are nonnegative. Their sum is
\begin{equation*}
  \sum_i\sum_{r=1}^{m_i}\frac{w_i\alpha_{ir}}{\sum_jw_j}
  =\frac{\sum_iw_i\sum_r\alpha_{ir}}{\sum_jw_j}
  =1.
\end{equation*}
Thus, the result is a finite Gaussian mixture with at most $\sum_i m_i$ component terms. If every component weight is positive and no two components have the same mean and covariance, no term is removed by zero-coefficient deletion or duplicate merging. The concatenated representation therefore has exactly $\sum_i m_i$ nonzero components.

\subsection*{Proof of Theorem~\ref{thm:compression-congruence}}

We prove the equivalence of the three statements.

Assume first that~\eqref{eq:incremental-compatibility} holds. We show by induction on a binary aggregation tree that every nontrivial subtree with leaf set $I$ evaluates to
\begin{equation*}
  C\!\left(\sum_{i\in I}M_i\right).
\end{equation*}
For two leaves $i,j$, the value is $M_i\odot_C M_j=C(M_i+M_j)$, so the claim holds. Consider a larger subtree with child leaf sets $I$ and $J$. If both children are nontrivial, the induction hypothesis gives the value
\begin{align*}
  C\!\left(C\!\left(\sum_{i\in I}M_i\right)+C\!\left(\sum_{j\in J}M_j\right)\right).
\end{align*}
Applying~\eqref{eq:incremental-compatibility} once with $M=\sum_{i\in I}M_i$ and $N=C(\sum_{j\in J}M_j)$ gives
\begin{equation*}
  C\!\left(\sum_{i\in I}M_i+C\!\left(\sum_{j\in J}M_j\right)\right).
\end{equation*}
By commutativity of addition and another use of~\eqref{eq:incremental-compatibility}, now with $M=\sum_{j\in J}M_j$ and $N=\sum_{i\in I}M_i$, this equals
\begin{equation*}
  C\!\left(\sum_{i\in I}M_i+\sum_{j\in J}M_j\right).
\end{equation*}
If one child is a leaf, the same argument uses~\eqref{eq:incremental-compatibility} once. Thus, every binary tree evaluates to final compression of the exact sum.

Next assume the tree condition. Since $0$ is the zero component measure in $\mathcal M$, it is an allowed leaf. Apply the tree condition to the three-leaf family $M,0,N$. In the tree that first combines $M$ and $0$, the first internal node gives
\begin{equation*}
  M\odot_C0=C(M+0)=C(M),
\end{equation*}
since $0$ is the additive zero measure. Combining with $N$ gives $C(C(M)+N)$. The tree condition says that every binary tree has the same value as final compression of the exact sum, namely $C(M+0+N)=C(M+N)$. Hence~\eqref{eq:incremental-compatibility} holds.

It remains to show the equivalence of incremental compatibility and congruence. Suppose first that $\sim$ is a congruence. Since $C$ is idempotent, $C(C(M))=C(M)$, so $M\sim C(M)$. Congruence under addition gives $M+N\sim C(M)+N$, which means
\begin{equation*}
  C(M+N)=C(C(M)+N).
\end{equation*}
Thus, incremental compatibility holds. Conversely, assume incremental compatibility and let $M\sim N$. Then $C(M)=C(N)$. For any $R\in\mathcal M$,
\begin{align*}
  C(M+R)&=C(C(M)+R)\\
        &=C(C(N)+R)\\
        &=C(N+R),
\end{align*}
where the first and last equalities use incremental compatibility. Hence $M+R\sim N+R$. Therefore, $\sim$ is a congruence for addition.

\subsection*{Proof of the transformed compression extension}

Let $g:\R_{>0}\to\R_{>0}$ be a strictly increasing bijection. For $M\neq0$, let
\begin{align*}
  \psi_g(M)&=\frac{g(W(M))}{W(M)}M\\
  \psi_g^{-1}(M)&=\frac{g^{-1}(W(M))}{W(M)}M,
\end{align*}
and let $\psi_g(0)=\psi_g^{-1}(0)=0$. Scalar multiplication gives
$W(\psi_g(M))=g(W(M))$ and
$W(\psi_g^{-1}(M))=g^{-1}(W(M))$. Consequently,
\begin{align*}
 \psi_g^{-1}(\psi_g(M))
 &=\frac{g^{-1}(g(W(M)))}{g(W(M))}
   \frac{g(W(M))}{W(M)}M=M,\\
 \psi_g(\psi_g^{-1}(M))
 &=\frac{g(g^{-1}(W(M)))}{g^{-1}(W(M))}
   \frac{g^{-1}(W(M))}{W(M)}M=M.
\end{align*}
Thus, $\psi_g$ is bijective. Let
\begin{equation*}
  M\boxplus_gN=\psi_g^{-1}\!\bigl(\psi_g(M)+\psi_g(N)\bigr).
\end{equation*}
The identity
\begin{equation*}
 \psi_g\bigl((M\boxplus_gN)\boxplus_gR\bigr)
 =\psi_g(M)+\psi_g(N)+\psi_g(R)
 =\psi_g\bigl(M\boxplus_g(N\boxplus_gR)\bigr)
\end{equation*}
and injectivity of $\psi_g$ prove associativity. Commutativity and the identity
$0$ follow from the corresponding properties of addition. Hence,
$(\mathcal M,\boxplus_g,0)$ is a commutative monoid.

Let $M\odot_C^gN=C(M\boxplus_gN)$. Assume the incremental condition
\begin{equation*}
 C\bigl(C(M)\boxplus_gN\bigr)=C(M\boxplus_gN).
\end{equation*}
For the exact aggregates $A$ and $B$ of two child subtrees,
\begin{align*}
 C\bigl(C(A)\boxplus_gC(B)\bigr)
 &=C\bigl(A\boxplus_gC(B)\bigr)\\
 &=C\bigl(C(B)\boxplus_gA\bigr)\\
 &=C(B\boxplus_gA)=C(A\boxplus_gB).
\end{align*}
Induction on the tree therefore gives final value
$C(M_1\boxplus_g\cdots\boxplus_gM_n)$ for every parenthesization. Conversely, applying the
tree condition to the leaves $M,0,N$ gives
\begin{equation*}
 C\bigl(C(M)\boxplus_gN\bigr)
 =C\bigl((M\boxplus_g0)\boxplus_gN\bigr)
 =C(M\boxplus_gN),
\end{equation*}
which is incremental compatibility.

If $\sim$ is a congruence, idempotence gives $M\sim C(M)$ and hence
$M\boxplus_gN\sim C(M)\boxplus_gN$, which is incremental compatibility.
Conversely, if incremental compatibility holds and $C(M)=C(N)$, then
\begin{align*}
 C(M\boxplus_gR)
 &=C(C(M)\boxplus_gR)\\
 &=C(C(N)\boxplus_gR)\\
 &=C(N\boxplus_gR).
\end{align*}
Hence, $M\boxplus_gR\sim N\boxplus_gR$, proving the congruence equivalence.

\subsection*{Verification of the pruning example}

Let $C$ retain the component with largest normalized coefficient, equivalently largest unnormalized component mass within a fixed total weight, preserve the total state weight, and break coefficient ties by choosing the smaller mean. For
\begin{equation*}
  x=(1,\N(0,1)),\qquad y=(1,\N(1,1))\quad\text{and}\quad z=(1,\N(2,1)),
\end{equation*}
exact fusion of $x$ and $y$ gives a two-component mixture with total weight $2$ and equal component coefficients. The tie rule keeps the smaller mean, so $C(x\oplus y)=(2,\N(0,1))$. Similarly, $C(y\oplus z)=(2,\N(1,1))$. Fusing the first compressed state with $z$ gives unnormalized component masses $2$ for $\N(0,1)$ and $1$ for $\N(2,1)$, so the largest-mass rule keeps $\N(0,1)$ and
\begin{equation*}
  (x\oplus_Cy)\oplus_Cz=(3,\N(0,1)).
\end{equation*}
Fusing $x$ with the second compressed state gives unnormalized component masses $1$ for $\N(0,1)$ and $2$ for $\N(1,1)$, so the largest-mass rule keeps $\N(1,1)$ and
\begin{equation*}
  x\oplus_C(y\oplus_Cz)=(3,\N(1,1)).
\end{equation*}
The two final states differ, so local pruning is not associative. The example is only a witness for tree dependence; it is not a claim about the statistical quality of pruning.

\end{document}